\newcommand{\BWT}{\ensuremath{\mathrm{BWT}}}
\newcommand{\SA}{\ensuremath{\mathrm{SA}}}
\newcommand{\LF}{\ensuremath{\mathrm{LF}}}
\newcommand{\pred}{\ensuremath{\mathrm{pred}}}
\newcommand{\rank}{\ensuremath{\mathrm{rank}}}
\newcommand{\select}{\ensuremath{\mathrm{select}}}
\newcommand{\floor}[1]{\left\lfloor #1 \right\rfloor}

\documentclass[a4paper,UKenglish,cleveref, autoref, thm-restate]{lipics-v2021}
\usepackage[htt]{hyphenat}
\title{RLBWT Tricks}

\author{Nathaniel K. Brown}{Faculty of Computer Science, Dalhousie University, NS, Canada}{nathaniel.brown@dal.ca}{https://orcid.org/0000-0002-6201-2301}{}
\author{Travis Gagie}{Faculty of Computer Science, Dalhousie University, NS, Canada}{travis.gagie@dal.ca}{https://orcid.org/0000-0003-3689-327X}{}
\author{Massimiliano Rossi}{Department of Computer and Information Science and Engineering, University of Florida, FL, USA}{rossi.m@ufl.edu}{https://orcid.org/0000-0002-3012-1394}{}

\authorrunning{N.\,K. Brown, T. Gagie and M. Rossi}

\Copyright{Nathniel K. Brown, Travis Gagie and Massimiliano Rossi}



\category{} 

\relatedversion{}
\relatedversiondetails{Full Version}{https://arxiv.org/abs/2112.04271}

\supplement{Source code available from \url{https://github.com/drnatebrown/r-index-f}}

\funding{This work was funded by NIH R01AI141810 and R01HG011392, NSERC Discovery Grant RGPIN-07185-2020, and NSF IIBR 2029552 and IIS 1618814.}

\acknowledgements{Many thanks to Omar Ahmed, Christina Boucher and Ben Langmead for discussions and assistance during our research, and to the anonymous reviewers for their insightful feedback.}

\nolinenumbers 

\hideLIPIcs  

\EventEditors{John Q. Open and Joan R. Access}
\EventNoEds{2}
\EventLongTitle{42nd Conference on Very Important Topics (CVIT 2016)}
\EventShortTitle{CVIT 2016}
\EventAcronym{CVIT}
\EventYear{2016}
\EventDate{December 24--27, 2016}
\EventLocation{Little Whinging, United Kingdom}
\EventLogo{}
\SeriesVolume{42}
\ArticleNo{23}

\begin{document}

\maketitle

\begin{abstract}
Until recently, most experts would probably have agreed we cannot backwards-step in constant time with a run-length compressed Burrows-Wheeler Transform (RLBWT), since doing so relies on rank queries on sparse bitvectors and those inherit lower bounds from predecessor queries.  At ICALP '21, however, Nishimoto and Tabei described a new, simple and constant-time implementation. For a permutation $\pi$, it stores an $O (r)$-space table --- where $r$ is the number of positions $i$ where either $i = 0$ or $\pi (i + 1) \neq \pi (i) + 1$ --- that enables the computation of successive values of $\pi(i)$ by table look-ups and linear scans. Nishimoto and Tabei showed how to increase the number of rows in the table to bound the length of the linear scans such that the query time for computing $\pi(i)$ is constant while maintaining $O (r)$-space. 

In this paper we refine Nishimoto and Tabei's approach, including a time-space tradeoff, and experimentally evaluate different implementations demonstrating the practicality of part of their result. We show that even without adding rows to the table, in practice we almost always scan only a few entries during queries. We propose a decomposition scheme of the permutation $\pi$ corresponding to the LF-mapping that allows an improved compression of the data structure, while limiting the query time. We tested our implementation on real-world genomic datasets and found that without compression of the table, backward-stepping is drastically faster than with sparse bitvector implementations but, unfortunately, also uses drastically more space.  After compression, backward-stepping is competitive both in time and space with the best existing implementations.
\end{abstract}

\section{Introduction}
\label{sec:introduction}

The FM-index~\cite{FM05} is the basis for key tools in computational genomics, such as the popular short-read aligners BWA~\cite{LD09} and Bowtie~\cite{LTPS09}, and is probably now the most important application of the Burrows-Wheeler Transform (\BWT)~\cite{BW94}.  As genomic databases have grown and researchers and clinicians have realized that using only one or a few reference genomes biases their results and diagnoses, interest in computational pan-genomics has surged and versions of the FM-index based on the run-length compressed \BWT\ (RLBWT)~\cite{MNSV10} have been developed that can index thousands of genomes in reasonable space~\cite{GNP20,KMBGLM20,ROLGB21}.  Those versions have all relied heavily on compressed sparse bitvectors, however, which are inherently slower than the bitvectors used in regular FM-indexes (see~\cite{Nav16} for details).  Experts would probably have guessed that sparse bitvectors were an essential component for RLBWT-based pan-genomic indexes --- until Nishimoto and Tabei~\cite{NT21} recently showed how to replace them with theoretically more efficient alternatives.

In particular, Nishimoto and Tabei's result gives an approach which achieves constant time \LF-mapping in $O(r)$-space \cite{NT21}. Speeding up \LF\ can reduce the time for basic queries over the RLBWT and other applications. For example, Ahmed et al.'s SPUMONI~\cite{Ahmed21} tool allows rapid targeted nanopore sequencing over compressed pan-genome indexes using approximate matching statistics; ``nontarget'' DNA molecules are ejected from the sequencer with an emphasis on speed. Their method depends on LF-mapping to extend matches, and otherwise ``jumping'' forwards or backwards in the \BWT\ based on threshold computation. Thresholds over the \BWT\ is a rather new approach, introduced by Bannai et al. in 2020 \cite{Bannai20}, suggesting further improvements may be developed; however, avoiding the lower bounds inherited from predecessor queries from rank on sparse bitvectors\footnote{Conventionally, \LF-mapping in runs bounded space relies on rank queries over sparse bitvectors.} is a more surprising result. For tools that heavily depend on \LF, experiments showing practical results provide an opportunity for speed improvements that otherwise would not have been expected to be attainable.

In this paper we focus on the first part of Nishimoto and Tabei's result: we demonstrate experimentally that we can reduce the time for basic queries on an RLBWT by replacing queries on sparse bitvectors by table lookups, sequential scans, and queries on relatively short uncompressed bitvectors. We implement \LF-mapping over the RLBWT using table lookup; preliminary results showed this could be made practical even without theoretical worst case time guarantees. Although their result also applies to the $\phi$ function over the RLBWT \cite{NT21}, we focus on \LF\ since it allows backward-stepping (performed before locating, which requires $\phi$) and its seems more compressible for \LF; we leverage the unique structure of \LF\ to partition columns of the table into non-decreasing subsequences.

With this motivation, we present various techniques and optimizations towards a practical implementation. To demonstrate its practicality, we use real-world genomic datasets to perform count queries using haplotypes of chromosome 19 and  SARS-CoV2 genomes. We find that our implementations are competitive in time/space with the best existing methods: in the average case without row insertions, and exploring a run splitting approach to loosely bound sequential scanning in the worst case. Further analysis shows in practice, sequential scans are quite rare, but can become more common as $n/r$ grows, motivating our run splitting and further approaches.

The rest of this paper is laid out as follows: in Section~\ref{sec:NT21} we present the two parts of Nishimoto and Tabei's result and explain how they relate to RLBWT-based pan-genomic indexes; Section~\ref{sec:implementation} describes methods used to make the result practical for implementation; in Section~\ref{sec:experiments} we present our experimental results; and in Section~\ref{sec:discussion} we analyse its practicality and summarize findings.

\section{Nishimoto and Tabei's Result}
\label{sec:NT21}

Suppose we want to compactly store a permutation $\pi$ on $\{0, \ldots, n - 1\}$ such that we can evaluate $\pi (i)$ quickly when given $i$.  If $\pi$ is chosen arbitrarily then $\Theta (n)$ space is necessary to store it in the worst case, and sufficient to allow constant-time evaluation.  If the sequence $\pi (0), \pi (1), \pi (2), \ldots, \pi (n - 1)$ consists of a relatively small number $b$ of unbroken incrementing subsequences, however --- meaning $\pi (i + 1) = \pi (i) + 1$ whenever $\pi (i)$ and $\pi (i + 1)$ are in the same subsequence --- then we can store $\pi$ in $O (b)$ space and evaluate it in $O (\log \log n)$ time.  To do this, we simply store in an $O (b)$-space predecessor data structure with $O (\log \log n)$ query time --- such as a compressed sparse bitvector --- each value $i$ such that $\pi (i)$ is the head of one of those subsequences, with $\pi (i)$ as satellite data; we evaluate any $\pi (i)$ in $O (\log \log n)$ time as
$$\pi (i) = \pi (\pred (i)) + i - \pred (i)\,.$$

Nishimoto and Tabei first proposed a simple alternative $O (b)$-space representation:\footnote{We may have taken some artistic license with their format.} we store a sorted table in which, for each subsequence head $p$, there is quadruple: $p$; the length of the subsequence starting with $p$; $\pi (p)$; and the index of the subsequence containing $\pi (p)$.

If we know the index of the subsequence containing $i$ then we can look up the quadruple for that subsequence and find both its head $p$ and $\pi (p)$, then compute $\pi (i) = \pi (p) + i - p$ in constant time.  If we want to compute $\pi^2 (i)$ the same way, however, we should compute the index of the subsequence containing $\pi (i)$, since $\pi (i)$ may be in a later subsequence than $\pi (p)$.  To do this, we look up the quadruple for the subsequence containing $\pi (p)$ (which takes constant time since we have its index) and find its head and length, from which we can tell if $\pi (i)$ is in the subsequence.  If it is not, we continue reading and checking the quadruples for the following subsequences (which takes constant time for each one, since they are next in the table) until we find the one that does contain $\pi (i)$.

Sequentially scanning the table to find the quadruple for the subsequence containing $\pi (i)$ could take $\Omega (b)$ time in the worst case, so Nishimoto and Tabei then proved the following result, which implies we can artificially divide some of the subsequences before building the table, such that all the sequential scans are short.  We still find their proof surprising, so we have included a summary of it below which introduces our parameter $d$. This refinement of the original theorem allows for a time/space tradeoff.

\begin{theorem}[Nishimoto and Tabei~\cite{NT21}]
\label{thm:NT21}
Let $\pi$ be a permutation on $\{0, \ldots, n - 1\}$,
$$P = \{0\} \cup \{i\ :\ 0 < i \leq n - 1, \pi (i) \neq \pi (i - 1) + 1\}\,,$$
and $Q = \{\pi (i)\ :\ i \in P\}$.  For any integer $d \geq 2$, we can construct $P'$ with $P \subseteq P' \subseteq \{0, \ldots, n - 1\}$ and $Q' = \{\pi (i)\ :\ i \in P'\}$ such that
\begin{itemize}
    \item if $q, q' \in Q'$ and $q$ is the predecessor of $q'$ in $Q'$, then $|[q, q') \cap P'| < 2 d$,
    \item $|P'| \leq \frac{d |P|}{d - 1}$.
\end{itemize}
\end{theorem}
\begin{proof}
We start by setting $P_0 = P$ and $Q_0 = Q$.  Suppose at some point we have $P_i$ and $Q_i = \{\pi (i)\ :\ i \in P_i\}$.  If there do not exist $q, q' \in Q_i$ such that $q$ is the predecessor of $q'$ in $Q_i$ and $|[q, q') \cap P_i| \geq 2 d$, then we stop and return $P' = P_i$ and $Q' = Q_i$; otherwise, we choose some such $q$ and $q'$.

We choose the $(d + 1)$st largest element $p$ in $[q, q') \cap P_i$ and set $P_{i + 1} = P_i \cup \{\pi^{- 1} (p)\}$ and $Q_{i + 1} = Q_i \cup \{p\} = \{\pi (i)\ :\ i \in P_{i + 1}\}$.  Since $q < p < q'$ we have $p \not \in Q_i$ and so $\pi^{- 1} (p) \not \in P_i$.  Therefore, $|P_{i + 1}| = |P_i| + 1$ and so, by induction, $|P_{i + 1}| = |P| + i + 1$.

Let $E_i$ be the set of intervals $[u, u')$ such that $u, u' \in Q_i$ and $u$ is the predecessor of $u'$ in $Q_i$ and $|[u, u') \cap P_i| \geq d$, and let $E_{i + 1}$ be the set of intervals $[u, u')$ such that $u, u' \in Q_{i + 1}$ and $u$ is the predecessor of $u'$ in $Q_{i + 1}$ and $|[u, u') \cap P_{i + 1}| \geq d$.  Since $E_{i + 1} = (E_i \backslash \{[q, q')\}) \cup \{[q, p), [p, q')\}$, we have $|E_{i + 1}| = |E_i| + 1$ and, by induction, $|E_{i + 1}| \geq i + 1$.

Since the intervals in $E_{i + 1}$ are disjoint and each contain at least $d$ elements of $P_{i + 1}$, we have $|P_{i + 1}| \geq d |E_{i + 1}| \geq d (i + 1)$.  Since $|P_{i + 1}| = |P| + i + 1$ and $|P_{i + 1}| \geq d (i + 1)$, we have $|P| + i + 1 \geq d (i + 1)$ and thus $i + 1 \leq \frac{|P|}{d - 1}$ and $|P_{i + 1}| = |P| + i + 1 \leq \frac{d |P|}{d - 1}$.  It follows that we find $P'$ and $Q'$ after at most $\frac{|P|}{d - 1}$ steps.
\end{proof}

To discuss how Theorem~\ref{thm:NT21} relates to RLBWTs, we first recall the definitions of the suffix array (\SA), the \BWT, the \LF\ mapping and $\phi$ for a text $T [0..n - 1]$:
\begin{itemize}
\item $\SA [i]$ is the starting position of the lexicographically $i$th suffix of $T$;
\item $\BWT [i]$ is the character immediately preceding that suffix;
\item $\LF (i)$ is the position of $\SA [i] - 1$ in \SA;
\item $\phi (i)$ is the value that precedes $i$ in \SA.
\end{itemize}
Let $T$ be defined over an alphabet $\Sigma$ of size $\sigma$. For convenience we assume $T$ ends with a special symbol $T [n - 1] = \mathtt{\$}$ that occurs nowhere else, we consider strings and arrays as cyclic and we work modulo $n$.

It is not difficult to see that \LF\ and $\phi$ (and thus also $\phi^{- 1}$) are permutations that can be divided into at most $r$ of unbroken incrementing subsequences, where $r$ is the number of runs in the \BWT.\footnote{Realizing this about $\phi$, however, led directly to Gagie, Navarro and Prezza's $r$-index~\cite{GNP20}.}  First, if $\BWT [i] = \BWT [i + 1]$ then $\LF (i + 1) = \LF (i) + 1$, so there are at most $r$ values for which $\LF (i + 1) \neq \LF (i) + 1$.  Second, if $\BWT [i] = \BWT [i + 1]$ so $\LF (i + 1) = \LF (i) + 1$ then
$$\SA [\LF (i)] = \phi (\SA [\LF (i + 1)])$$
and, as illustrated in Figure~\ref{fig:equality},
$$\phi (\SA [i + 1])
= \SA [i]
= \SA [\LF (i)] + 1
= \phi (\SA [\LF (i + 1)]) + 1
= \phi (\SA [i + 1] - 1) + 1$$
or, choosing $i' = \SA [i + 1] - 1$, we have $\phi (i' + 1) = \phi (i') + 1$.  It follows that there are at most $r$ values for which $\phi (i' + 1) \neq \phi (i') + 1$.  Nishimoto and Tabei's result therefore gives us $O (r)$-space data structures supporting \LF, $\phi$ and $\phi^{- 1}$ in constant time.

\begin{figure}[t]
\begin{center}
\includegraphics[width=0.5\textwidth]{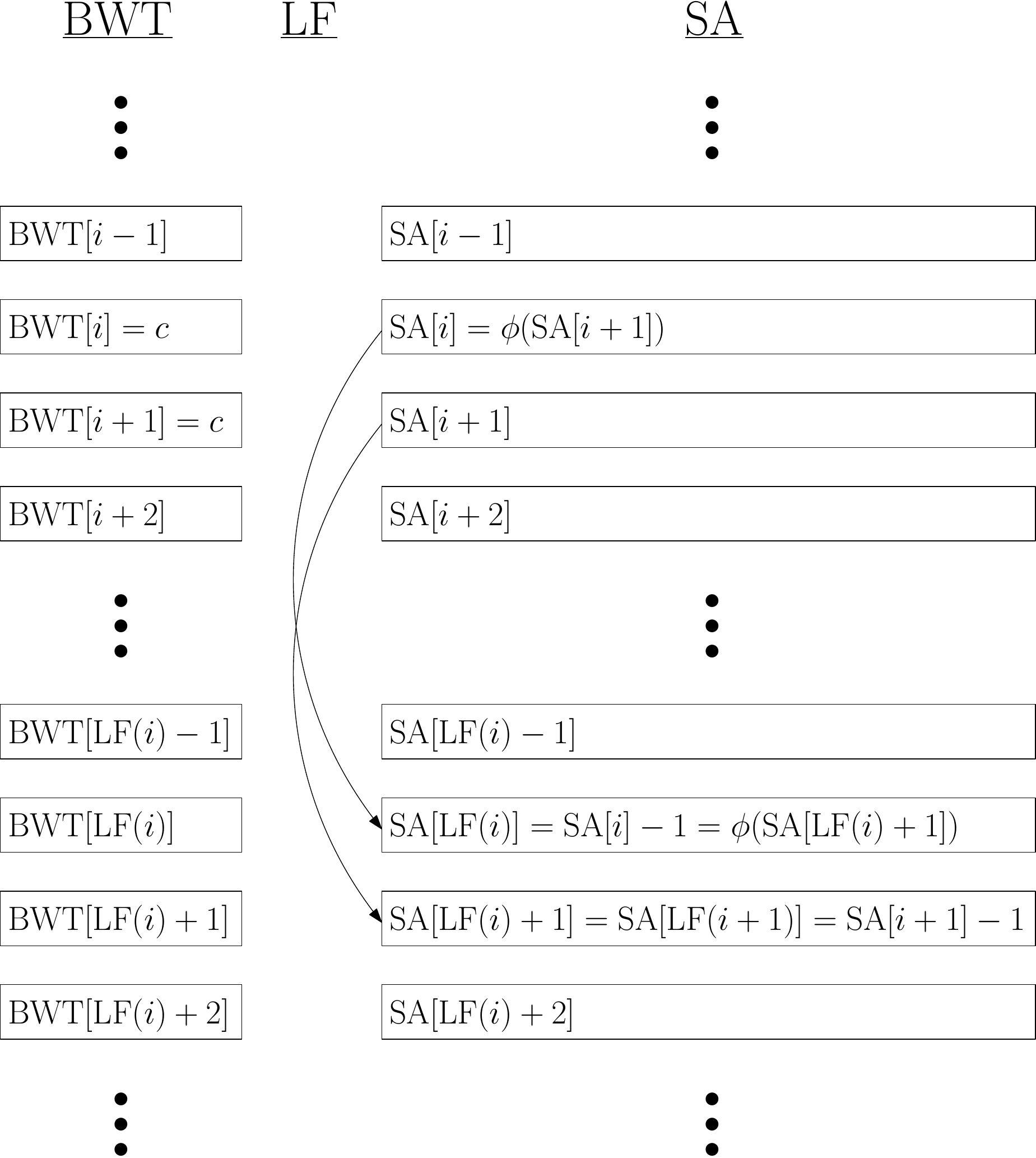}
\caption{An illustration of why $\BWT [i] = \BWT [i + 1]$ implies $\phi (\SA [i + 1]) = \phi (\SA [i + 1] - 1) + 1$.}
\label{fig:equality}
\end{center}
\end{figure}

As a practical aside we note that, although applying Theorem~\ref{thm:NT21} means we store quadruples for sub-runs in the \BWT, we can store with them the indexes of the maximal runs containing them and thus, for example, store \SA\ samples in an $r$-index only at the boundaries of maximal runs and not sub-runs.

The queries needed for most RLBWT-based pan-genomic indexes\footnote{For example, for the recent pan-genomic index MONI~\cite{ROLGB21}, we need \LF, $\phi$, $\phi^{- 1}$ and access to so-called thresholds.  A threshold for a consecutive pair of runs of the same character in \BWT\ is a position of a minimum LCP value in the interval between those runs.  If we know the index of the run containing a particular character $\BWT [i]$ and its offset in that run, and we want to know whether it is before or after the threshold for the pair of runs of another character $c$ bracketing $\BWT [i]$, then we can find in $O (\log \sigma)$ time the index of the preceding run of $c$s; if we have the index of the run containing the threshold and its offset in that run stored with that preceding run of $c$s, then we can tell immediately if $\BWT [i]$ is before or after the threshold.} can be implemented using \LF, $\phi$, $\phi^{- 1}$ and access, \rank\ and \select\ queries on the string $R [0..r - 1]$ in which $R [i]$ is the distinct character appearing in the $i$th run in \BWT, which can be supported with a wavelet tree on $R$.  Of course that wavelet tree uses bitvectors, but even with uncompressed bitvectors it takes only $r \lg \sigma + o (r \log \sigma)$ bits, where $\sigma$ is the size of the alphabet (usually 4 for genomics and pan-genomics), and supports those queries in $O (\log \sigma)$ time (or constant time when $\sigma = \log^{O (1)} n$).

\section{Practical Approach}
\label{sec:implementation}

To provide a practical implementation of Nishimoto and Tabei's first result, we slightly modify the structure of the table. Consider the permutation to be \LF(i)\ over the \BWT, with runs being unary substrings of the \BWT. In Section~\ref{sec:NT21} we presented the quadruples using absolute indexes over the permutation, but we can instead perform access using the run index itself: let positions of run heads in the \BWT ~be the array $I[0..r-1]$ storing the sorted values $i$ such that $i = 0$ or $\LF (i - 1) \neq \LF (i) - 1$. For all $k\in \{0,1,\dots,|I|-1\}$ we store a triple containing: the length of the run, i.e. $I[k+1] - I[k]$, where $I[r] = n$; the index of the run containing $\LF(I[k])$, i.e., $\max \{j \mid I[j] \leq \LF(I[k])\} $; and the offset $d$ of $\LF(I[k])$ in that run. Let $j$ be a position in the $k$-th run, the offset of $\LF (j)$ and $\LF (I[k])$ is $I[k] - j$, hence we can find the correct run containing $\LF (j)$ and its offset in that run using a sequential scan as described in Section~\ref{sec:NT21}. With this approach, we can represent positions in the \BWT~as run/offset pairs and implement \LF\ accordingly, i.e. $(k',d') \gets \LF (k, d)$. This change removes the need for the $p$ column of the table, with successive \LF\ steps  performed using the returned run/offset pair; access row $k'$ with offset $d'$ and perform $\LF (k', d')$.

\subsection{Block Compression}
\label{sec:implementation:block}

For each row on the previous representation of the \BWT, we store the character of the run corresponding to the row to enable support of count and inversion queries. Figure~\ref{fig:table_example} shows an example of this uncompressed table. Preliminary results showed that left uncompressed, \LF-mapping could be made drastically faster than a sparse bitvector implementation (seen in Section~\ref{sec:experiments} as {\texttt {rle-string}}) for inversion or \LF\ queries. However, the result is also drastically larger; this formulation is not practical because it requires storing three integers and one character for each run, and to perform count operations, it requires scanning the run heads to find the preceding and following run of the character we are seeking. One first improvement is to store the array $R[0..r-1]$ in a wavelet-tree as described in Section~\ref{sec:NT21}, which supports \rank\ and \select\ queries to efficiently find the preceding and following run of a given character. 

\begin{figure}[ht]
\begin{center}
\includegraphics[width=0.8\textwidth]{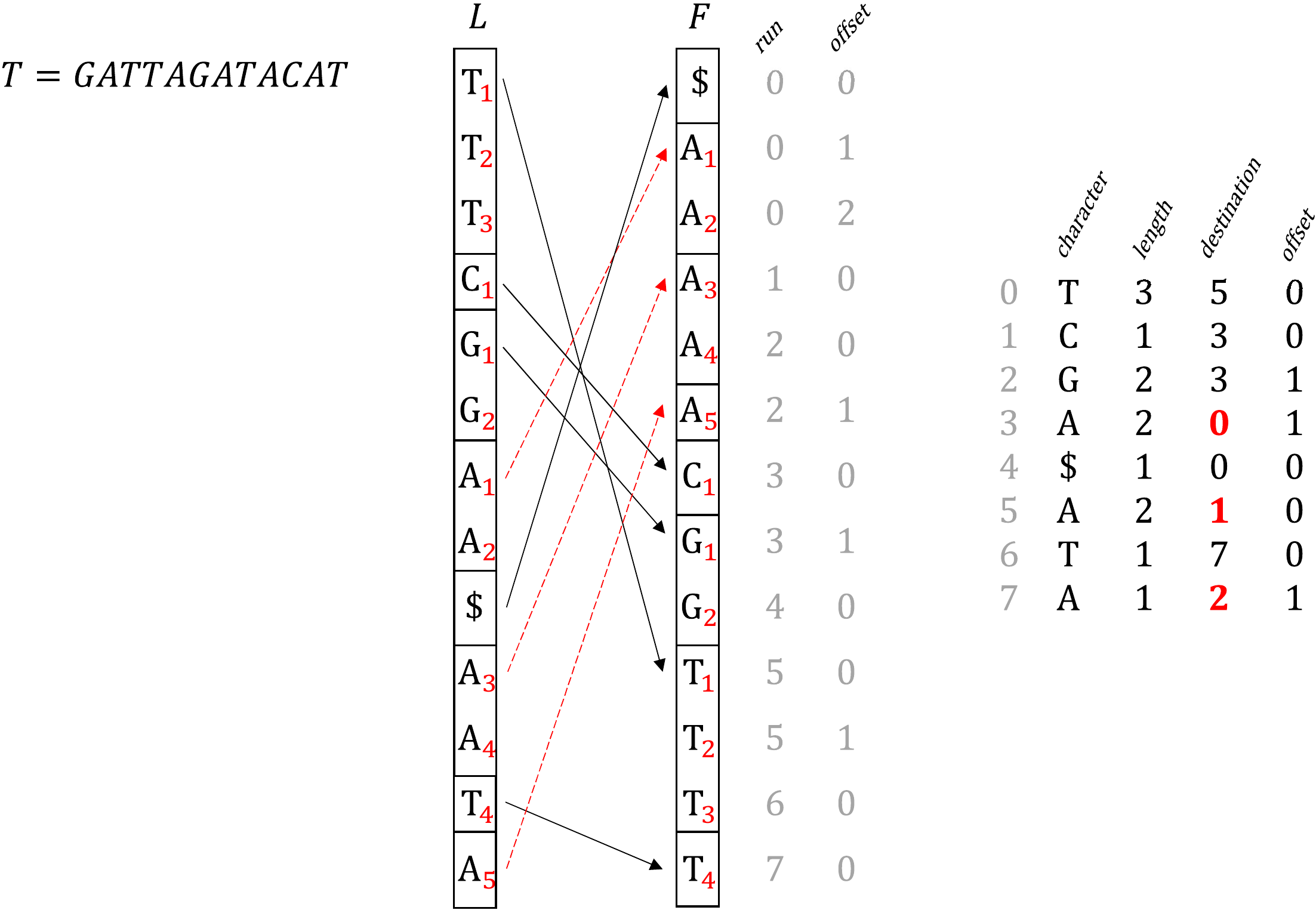}
\caption{For an example text $T = GATTAGATACAT$, the \LF\ mapping and subsequent uncompressed table is built (with appended terminal character \$). The run/offset columns show positions with respect to the L column used to find a mappings predecessor. Notice that highlighted stored mappings (destinations) for any run of $A$s form a non-decreasing subsequence.}
\label{fig:table_example}
\end{center}
\end{figure}

The tabular approach exploits space locality of the entries that facilitate the linear scans required by the algorithm when accessing rows sequentially; however, there is no apparent relationship which makes row-wise compression easy. To mitigate locality concerns, we partition the table into blocks of size $B$ which are loaded in a cache friendly manner. Using a fixed $B$, we can easily perform modular arithmetic to map positions within the blocks. For each block we store the corresponding character of each run in a wavelet tree that allows fast \rank\ and \select\ queries inside the block (using uncompressed bitvectors). For each character $c$ of the alphabet, the position of the first run of $c$'s preceding the beginning of the block and following the end of the block is stored, allowing efficient retrieval of these characters' correct rows when they are not stored in the wavelet tree and occur in another block. For example, we may need to look to another block if some character has no occurrences in the current block, or has no occurrences before/after some position.

To improve compression inside the block, we compress the lists of lengths and offsets using directly-addressable codes (DACs, see~\cite{Nav16}); we divide the list of run indices into $\sigma$ sub-lists, each containing the indices from rows corresponding to runs of a distinct character $c \in \Sigma$. Compressing the lengths and offsets in DACs is naive compression\footnote{DACs are a simple method to allow both random access alongside compression; however, more specific techniques would be preferred if these columns have exploitable properties that we could not uncover.} leveraging the length of a value's bit representation while also supporting random access. For mapping destinations, it follows from \LF\ that the mapping indices across a common character $c$ form a non-decreasing sub-sequence \cite{BW94} as highlighted in Figure~\ref{fig:table_example}. If we store in a block, for each of the $\sigma$ sub-lists, the mapping index of the first occurrence in the list, then the rest of the list can be truncated as a difference from the base mapping. We can also choose to represent the sub-lists by partial differences; for $m$ occurrences of a character $c$ let $M[0..m-1]$ be such a sub-list where we explicitly store the first mapping $M[0]$, and represent the list as $D[0] = 0, D[i] = M[i] - M[i-1]$. Storing only partial differences allows us to recover the mapping using prefix sum, which we expand upon in Section~\ref{sec:implementation:opt} alongside an approach over absolute offsets from the base. To manoeuvre around our positional change to run indices, we also store a sparse bitvector marking sampled run head positions in the \BWT, which is used after backwards-stepping to recover the absolute index from a run/offset pair.\footnote{Although we introduce a sparse bitvector into our data structure, it is not used during sequential \LF\ stepping, but rather as an ``exit'' or ``entrance'' from the table's run/offset pairs.}

\subsection{Optimizations}
\label{sec:implementation:opt}
Compressing the mapping column as ``difference lists'' gives various representations of exploiting the $\sigma$ non-decreasing sub-sequences: 
\begin{description}
\item[DAC Sampling]
By storing the partial differences space efficiently and sampling the absolute difference from the base, the number of random accesses needed to recover the correct value is bounded when computing the prefix sum. Implementing the approach using DACs to store the partial differences, we have a first method to retrieve mappings in compressed space while avoiding a costly traversal of the entire list. Although basic, this method is a simple choice to illustrate how we can leverage these sequences being non-decreasing.
\item[Linear Interpolation]
We perform linear interpolation between sampled offsets (as opposed to partial differences); with a sample rate $s$, prior sample $x$, next sample $z$, and unsampled difference $y$ at position $i$. For each $y$, we then store its difference $\Delta = y - \epsilon$ from a weighted average defined by $$\epsilon = x + (z - x) \cdot (i - s\cdot \left( \floor{i/s})/s\right)$$ into a DAC. \footnote{We store a bitvector denoting the sign of the stored component, allowing us to compress unsigned integers using the DAC.} Given $i$ and $s$, we lookup $x$ and $z$ to compute $\epsilon$, after which we compute $(y - \epsilon) + \epsilon = y$ from our stored value $\Delta = y - \epsilon$ to recover the mapping. At worst the stored value can only be the difference between the sampled values themselves, and we expect each value to tend towards the interpolated average obtained by assuming a linear increase between samples.
\item[Bitvector]
Construct an uncompressed bitvector in which the number of 0s before the $k$th 1 is the offset from the first pointer (which is stored explicitly) to the $k$th. For example, given a sequence $M = [11, 16, 19, 21]$ and corresponding partial differences $D = [0, 5, 3, 2]$, we store the first pointer $M[0] = 11$ alongside the bitvector $$10000010001001$$ constructed as described above. Performing $\select(k) - k$ over this bitvector returns the number of 0s prior to the $k$th 1 and recovers the difference; in essence, a prefix sum over the partial differences where we remove the $k$ number of 1s from our calculation. Adding the stored $M[0]$ to this difference restores the original value. Given our example and $k = 3$, we have $$M[0] + \select(3) - 3 = 11 + 11 - 3 = 19 = M[2]$$ and we recover the correct value at $M[2]$ ($i = 2$ corresponds to the $k = 3$ bit due to 0-based array indexing).
\end{description}

To further optimize for practical input, consider an alternative to the wavelet tree suitable for small alphabets or when query support is needed for only a subset of characters. Where the wavelet tree performs \rank\ and \select\ over multiple tree levels, we could instead store full length uncompressed bitvectors in our blocks, one for each chosen character $c$ marking positions $i$ where $\BWT [i] = c$. For large alphabets, this approach is much larger than a wavelet tree representation; however, for genomic datasets which in practice support queries on few characters such as the nucleobases $\{A, C, G, T\}$, this alternative may be preferred. As this is the case in our experiments, we use this restricted alphabet trick to trade off space for increased speed in performing \rank\ and \select\ operations. A summary of the structure of our proposed practical approach is shown in Figure~\ref{fig:implementation}; an overview of the hierarchy of the proposed optimizations with respect to components of the data structure and the varying options which we have implemented.

\begin{figure}[ht]
\begin{minipage}{1\linewidth}
\includegraphics[width=\textwidth]{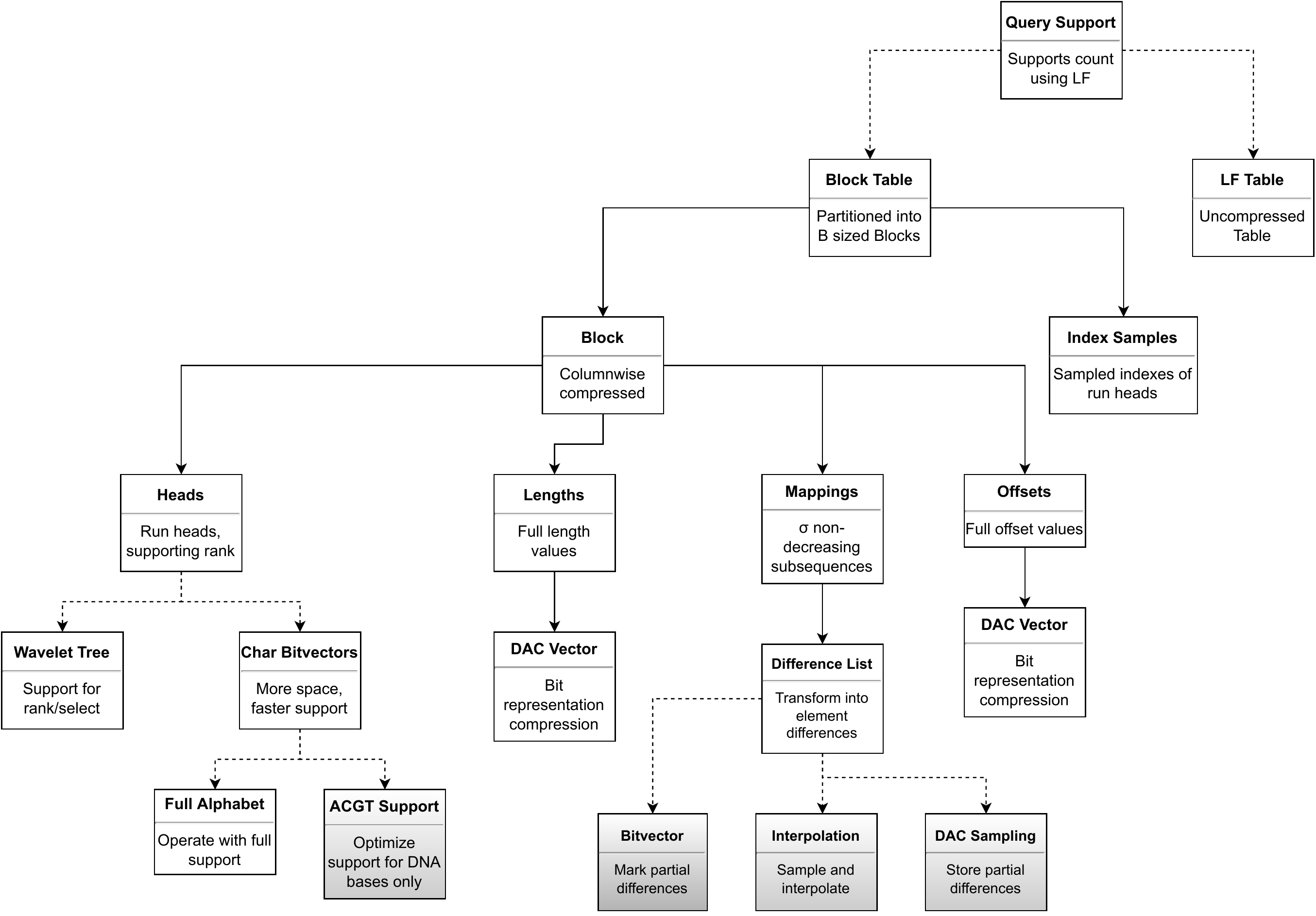}
\end{minipage}
\caption{Shows hierarchy of implementation, outlining different approaches and optimizations. Solid lines show required
components of parents given our work, where dotted lines denote multiple options being available. For example, the various methods to recover the mapping of a run head are shown as children of difference list. Shaded nodes show paths that are implemented for experiments in Section~\ref{sec:experiments}.
\label{fig:implementation}}
\end{figure}

\subsection{Scanning Complexity}

\label{sec:implementation:constructions}
We have not yet implemented the second part of Nishimoto and Tabei's result because we correctly expected their idea of table lookup (perhaps modified slightly) to be interesting and practical by itself. Over real world datasets (as discussed in Section~\ref{sec:discussion}), our typical sequential scan is very small; however, theoretically we use $\Omega (r)$-time in the worst case for such a scan for \LF\ .  In fact, there are strings for which the average time for a scan is $\Omega (r)$.  Suppose a string has $\BWT[0..n-1] = (bc)^{n/10}\cdot(a)^{4n/5}$ with $r = n / 5 + 1$ runs. By \LF\ properties we have $\frac{3n}{5}$ \LF\ steps which require scanning $r-1$ rows, as described in Figure~\ref{fig:amortized}. Similarly, we encounter $\Omega (n \cdot r)$-time for inversion, as we perform exactly $n$ possible \LF\ steps during a full retrieval of the original string.

\begin{figure}[ht]
  \begin{minipage}[c]{0.4\textwidth}
    \includegraphics[width=\textwidth]{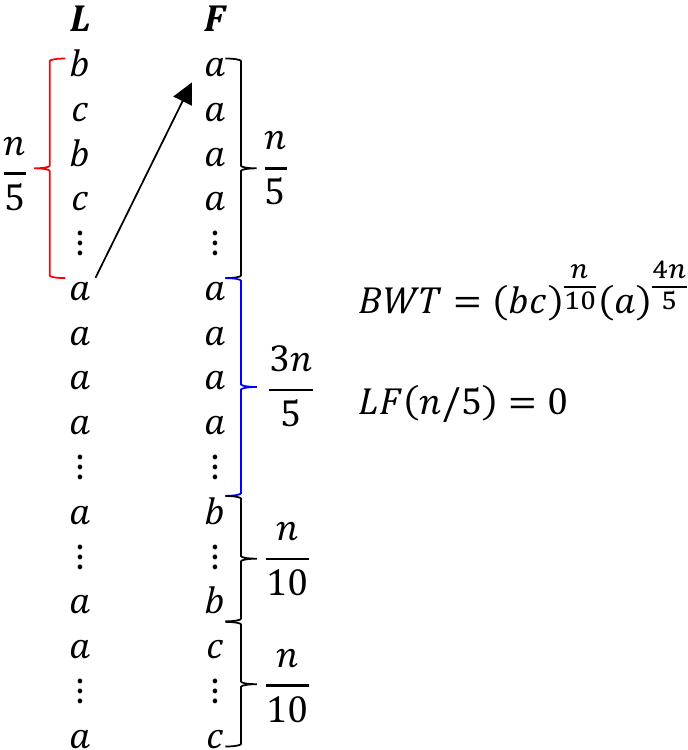}
  \end{minipage}\hfill
  \begin{minipage}[c]{0.55\textwidth}
    \caption{Visual representation of amortized analysis in Section~\ref{sec:implementation:constructions}. Notice that given a \BWT\ of this form, any character $a$ corresponding to run $k$ with $I[k] = n/5$ stores $\LF (n/5) = 0$ as its mapping. If the offset $d$ is greater than $n/5$, then the sequential scan must cross the boundaries of each run of $b$ or $c$, of which there are $n/5$ in total; since there is only one run of $a$, we scan $r-1$ entries, and perform this operation for $\frac{3n}{5}$ possible steps. Amortized over all possible \LF\ steps, we cannot avoid $\Omega (r)$ scans in the worst case.
    \label{fig:amortized}}
  \end{minipage}
\end{figure}

In practice, a very similar string can be produced which preserves a similar worst case. Consider a randomly generated binary string, for our purposes over the alphabet $\Sigma = \{b, c\}$. We then interleave the sequence with four consecutive $a$ characters between each of the original characters (resulting in $\frac{4n}{5}$ $a$ characters). The number of runs we expect in its \BWT\ cannot be much less than the number of runs in the original sequence, since the introduced $a$ characters are easily run-length compressed and such a technique would improve compression of any random sequence otherwise. The expected number of runs in a random binary string is half its length ($\frac{n}{10}$), also observed in practical experiments, and thus mapping to $a$ characters results in almost the same case as Figure~\ref{fig:amortized}.  To perform some practical bounding against scans without theoretical guarantees, we allow splitting of large runs by specifying the maximum acceptable run length to provide an alternative construction.

\subsection{Count Queries}
Standard FM-indexes are particularly good at counting queries, both in theory and in practice, and counting was also the first query supported quickly and in small space by RLBWT-based versions of the FM-index~\cite{MNSV10} (time- and space-efficient reporting was developed much later~\cite{GNP20}).  It seems appropriate, therefore, to test with counting queries our implementation of the first part of Nishimoto and Tabei's result. A counting query for pattern $S [0..m - 1]$ in text $T [0..n - 1]$ returns the number of occurrences of $S$ in $T$, by backward searching for $S$ and returning the length of the \BWT\ interval containing the characters preceding occurrences of $S$ in $T$.  We can implement a backward step using access to the string $R$ described in Section~\ref{sec:NT21}, up to 2 \rank\ queries and 2 \select\ queries on $R$, and 2 \LF\ queries.

Suppose the interval $\BWT [s..e]$ contains the characters preceding occurrences of $S [i + 1..m - 1]$ in $T$ and we know both the indices $j_s$ and $j_e$ of the runs containing $\BWT [s]$ and $\BWT [e]$, and the offsets of those characters in those runs.  We need not assume we know $s$ and $e$ themselves.  If $R [j_s] = S [i]$ then the \BWT\ interval containing the characters preceding occurrences of $S [i..m - 1]$ in $T$, starts at $\BWT [\LF (s)]$.  Otherwise, it starts at $\BWT [\LF (s')]$, where $s'$ is the first character in run 
$$j_{s'} = R.\select_{S [i]} (R.\rank_{S [i]} (j_s) + 1)\,.$$
Symmetrically, if $R [j_e] = S [i]$ then the interval ends at $\BWT [\LF (e)]$; otherwise, it ends at $\BWT [\LF (e')]$, where $e'$ is the last character in run
$$j_{e'} = R.\select_{S [i]} (R.\rank_{S [i]} (j_e))\,.$$ (If $j_{s'} > j_{e'}$ then $S [i..m - 1]$ does not occur in $T$.) These operations are all supported across our block compressed table, and the final interval positions in the \BWT\ can be computed using sampled run head positions to return the final count.

\section{Experiments}
\label{sec:experiments}

Our code was written in C++ and compiled with flags \texttt{-O3} \texttt{-DNDEBUG} \texttt{-funroll-loops} \texttt{-msse4.2} using data structures from \texttt{sdsl-lite}~\cite{GBMP14}.  We performed our experiments on a server with an Intel$^{\mbox{\tiny \textregistered}}$ Xeon$^{\mbox{\tiny \textregistered}}$ Silver 4214 CPU running at 2.20GHz with 32 cores and 100 GB of memory.  Our code is available at \texttt{\url{https://github.com/drnatebrown/r-index-f.git}}. Count query times were measured using Google Benchmark, and construction with the Unix {\texttt /usr/bin/time} command.

\subsection{Data Structures}
For our table lookup implementations, we partition into blocks of size $B=2^{20}$ and sample every 16th run position in the BWT. We compared the following data structures:
\begin{description}
\item[lookup-bv] table lookup with bitvector marking differences with 0s, recovered with \select\ described in Section~\ref{sec:implementation:opt}.
\item[lookup-int] table lookup with linear interpolation between sampled values described in Section~\ref{sec:implementation:opt} with sample rate 16.
\item[lookup-dac] table lookup with DAC sampling of differences described in Section~\ref{sec:implementation:opt} with sample rate 5.
\item[lookup-split2] table lookup with naive run splitting using \textit{lookup-bv} data structure described in Sections~\ref{sec:implementation:opt}, \ref{sec:implementation:constructions}. Runs larger than twice the average length $n/r$ are split.
\item[lookup-split5] table lookup identical to \textit{lookup-split2}, except runs larger than five times the average length $n/r$ are split.
\item[wt-fbb] fixed-block boosting wavelet tree of~\cite{fasterminuter} using default parameters; implementation at \texttt{\url{https://github.com/dominikkempa/faster-minuter}}.
\item[rle-string] run-length encoded string of the $r$-index~\cite{GNP20}; implementation based off \texttt{\url{https://github.com/nicolaprezza/r-index}}.
\item[RLCSA] the \BWT\ component\footnote{We build the data structure without suffix-array sampling.} of the run-length encoded compressed suffix array of~\cite{MNSV10} using default parameters; implementation at \texttt{\url{https://github.com/adamnovak/rlcsa}}.
\end{description}

\subsection{Datasets}
We tested our data structures for construction and query on 4 collections of 128, 256, 512 and 1000 haplotypes of chromosome 19 from the 1000 Genomes Project~\cite{1000genomes} ({\texttt {chr19}}) and 4 collections of 100k, 200k, 300k, 400k SARS-CoV2 genomes from the EBI's COVID-19 data portal~\cite{ebi21}\footnote{The complete list of accession numbers is reported in the repository.} ({\texttt {Sars-CoV2}}). Each set is a superset of the previous one. Table~\ref{tab:realdatasets} describes the lengths $n$ and ratio $n/r$ of the datasets.

\begin{table}[!ht]%
	 		\centering
		 		\begin{tabular}{llrrr}
		 			\hline
		 			Name & Description & {$N$} & {$n/10^6$} & {$n/r$} \\ 
		 			\hline
		 			{\tt chr19} & Human chromosome 19  & 128 &  ~7568.01 & 222.24 \\ 
 		 			{\tt chr19} & Human chromosome 19  & 256 &  ~15136.04 & 424.93 \\ 
  		 			{\tt chr19} & Human chromosome 19  & 512 &  ~30272.08 &  771.54\\ 
   		 			{\tt chr19} & Human chromosome 19  & 1,000 &  ~59125.12 & 1287.38 \\ 
		 			{\tt Sars-CoV2} & Sars-CoV2 genomes database & 100,000  & ~2979.01 & 881.16 \\ 
 		 			{\tt Sars-CoV2} & Sars-CoV2 genomes database & 200,000  & ~5958.35 & 977.19 \\ 
  		 			{\tt Sars-CoV2} & Sars-CoV2 genomes database & 300,000  & ~8944.37 & 1178.00 \\ 
   		 			{\tt Sars-CoV2} & Sars-CoV2 genomes database & 400,000  & ~11931.17 & 1328.92 \\ 
		 			\hline
		 		\end{tabular}
		 		\bigskip
	 		\caption{Table of the different datasets.  In column 1 and 2 we report the name and description of the datasets, in column 3 we report the number of sequences in the collection, in column 4 we report the length of the file, and in column 5 the ratio of the length to the number of runs in the \BWT.  \label{tab:realdatasets}}
	 	\end{table}

\subsection{Construction}
In Figure~\ref{fig:construction} we report the time and memory for construction of the data structures for the {\texttt {chr19}} and {\texttt{Sars-CoV2}} datasets. {\texttt {RLCSA}} is omitted, since it is the only data structure not built using prefix free parsing (PFP)~\cite{PFP19}, and its construction time far exceeded the other methods.

\begin{figure}[!htb]
    \centering
    \begin{minipage}{.5\textwidth}
            \centering
            \includegraphics[width=\textwidth]{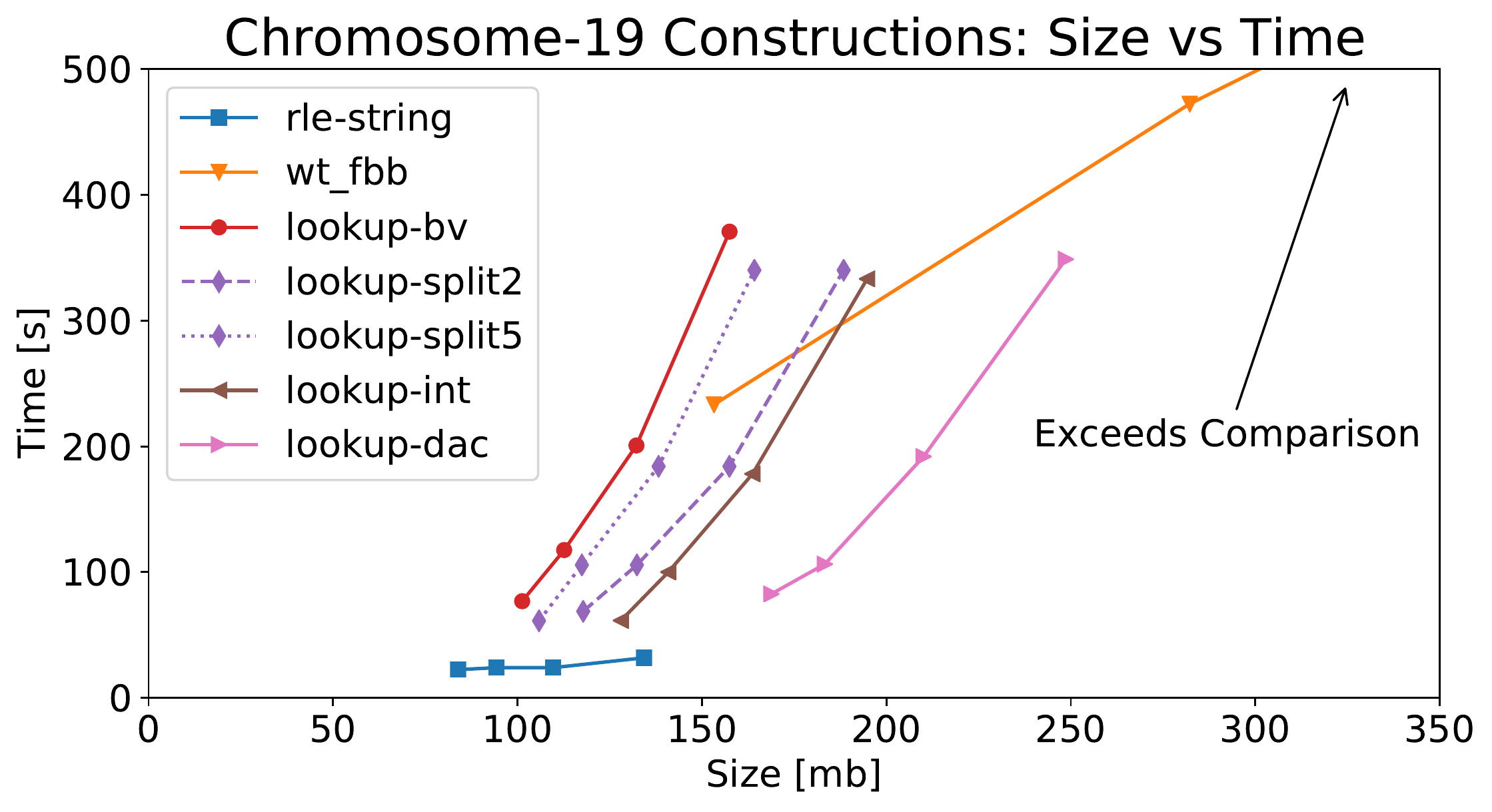}
        \end{minipage}%
        \begin{minipage}{0.5\textwidth}
            \centering
            \includegraphics[width=\textwidth]{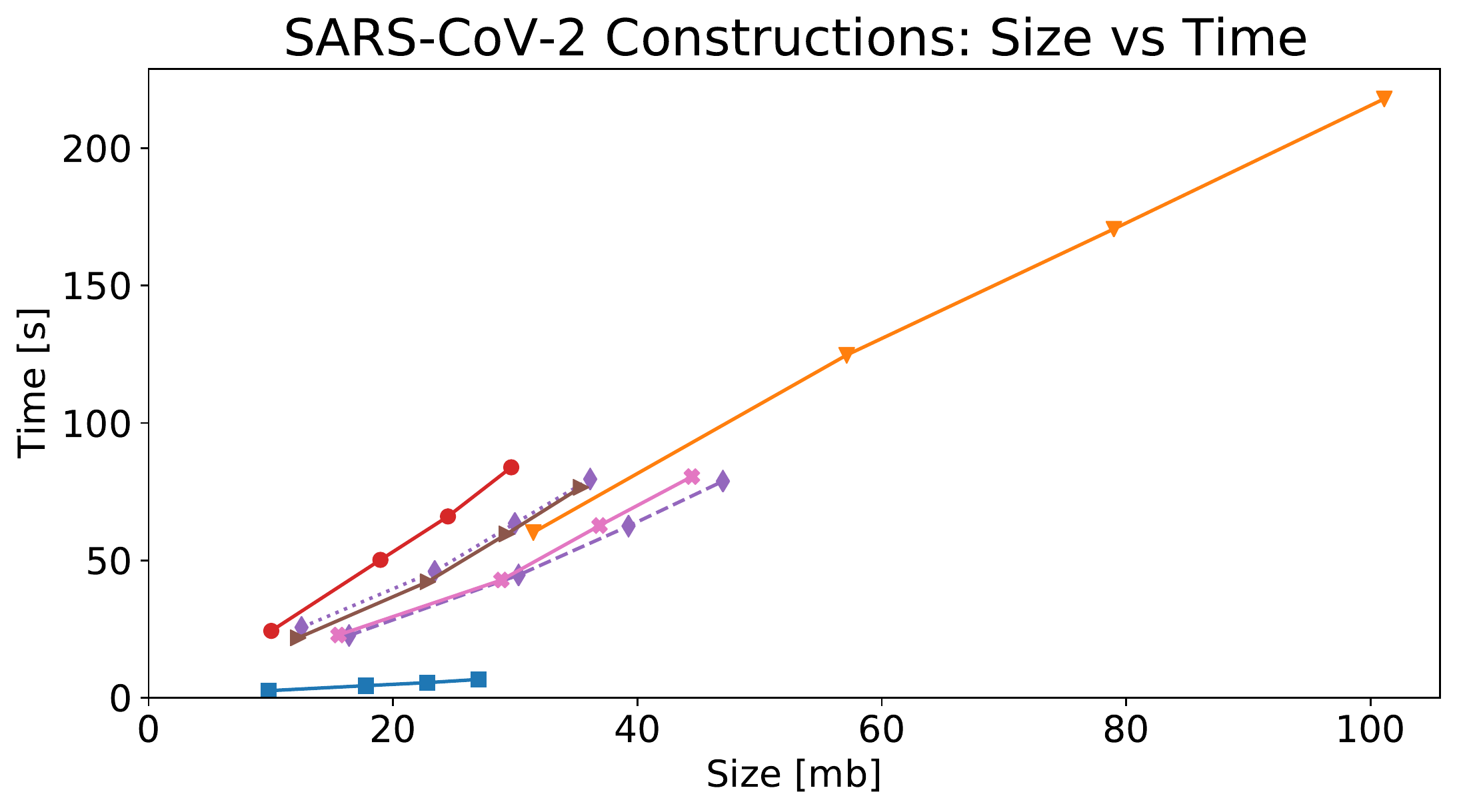}
        \end{minipage}\textbf{}
        \caption{Construction for {\texttt {chr19}} of 128, 256, 512 and 1000 copies (left) and {\texttt {Sars-CoV2}} of 100k, 200k, 300k and 400k copies (right). Copies increase for an instance plotted left to right. For {\texttt {chr19}} we partially omit {\texttt {wt-fbb}} for being magnitudes larger than other values (approximately 4 times slower and larger than {\texttt {lookup-bv}} for 512 copies and similarly 5 times slower and 7 times larger for 1000).
        \label{fig:construction}}
\end{figure}

\subsection{Query}
To query the data structures we performed counting queries for 10000 randomly chosen substrings each of length 10, 100, 1000 and 10000. In Figure \ref{fig:experiments_chr19} and \ref{fig:experiments_sars} we report the time and memory for querying of the data structures for the {\texttt {chr19}} and {\texttt{Sars-CoV2}} datasets respectively.

\begin{figure}[!p]
\begin{minipage}{1\linewidth}
\includegraphics[width=\textwidth]{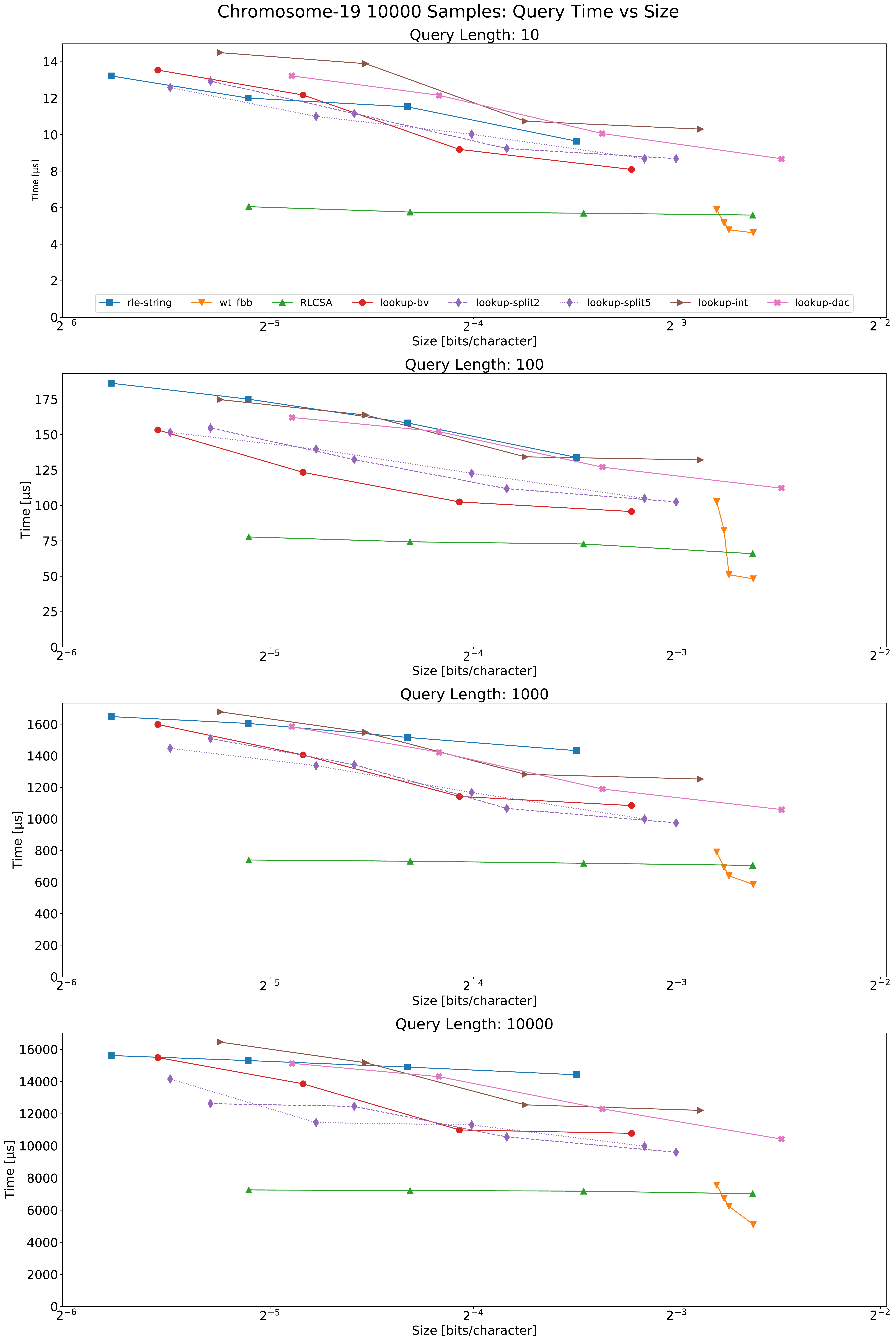}
\end{minipage}
	\caption{The time per query to count the occurrences of 128, 256, 512 and 1000 copies of {\texttt{chr19}} for 10000 randomly-chosen substrings of length 10, 100, 1000 and 10000 each. Copies for a single line are read from largest number of copies to smallest, left to right. The x axis is logarithmically scaled, motivated by doubling the number of copies across examples.
	\label{fig:experiments_chr19}}
\end{figure}%

\begin{figure}[!p]
\begin{minipage}{1\linewidth}
\includegraphics[width=\textwidth]{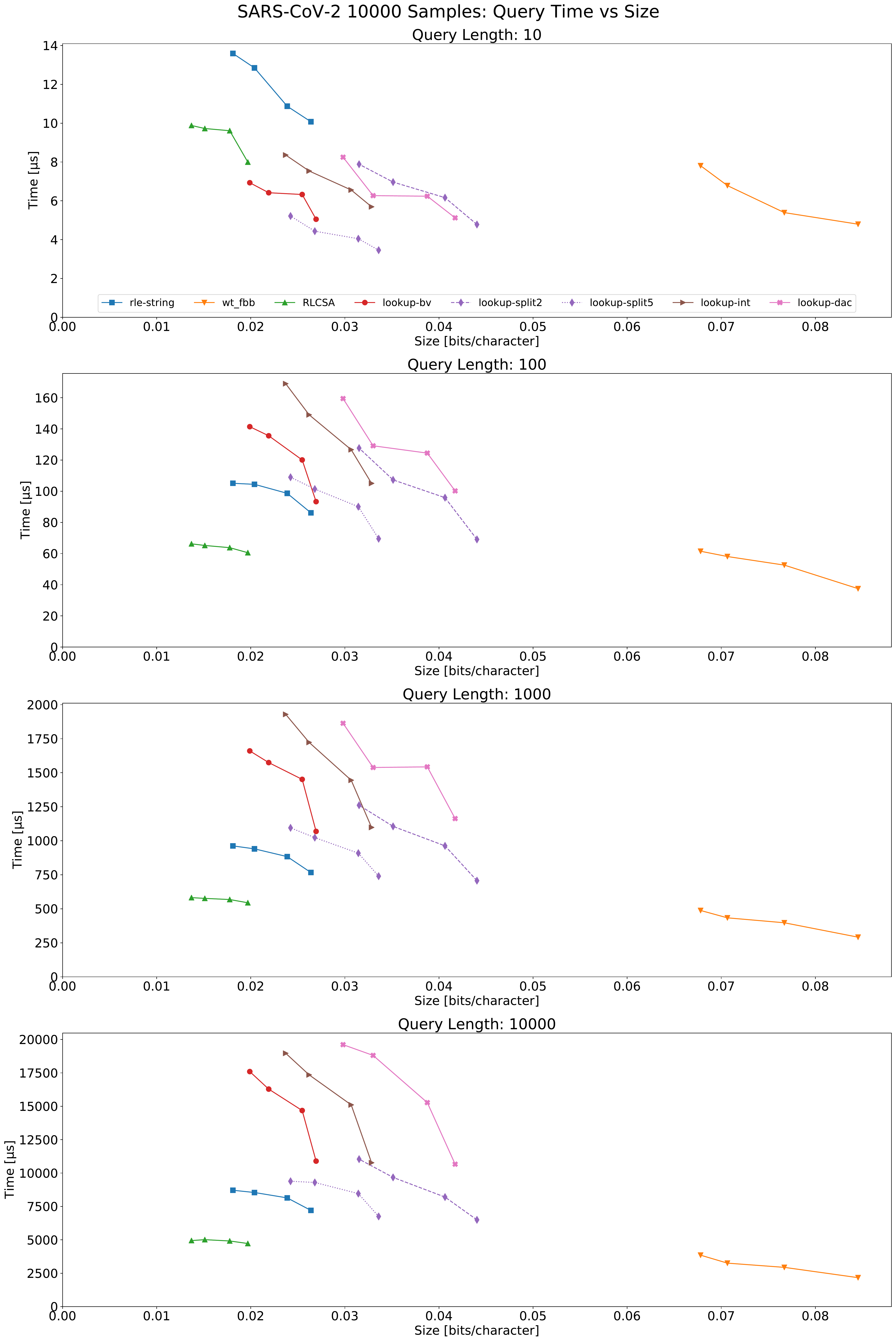}
\end{minipage}
	\caption{The time per query to count the occurrences of 100k, 200k, 300k and 400k {\texttt{Sars-CoV2}} copies for 10000 randomly-chosen substrings. Results are given for queries of length 10, 100, 1000 and 10000. Copies for a single line are read from largest number of copies to smallest, left to right.
	\label{fig:experiments_sars}}
\end{figure}%

\section{Discussion}\label{sec:discussion}

With respect to our table lookup implementations, {\texttt {lookup-bv}} and its variants ({\texttt {lookup-split2}}, {\texttt {lookup-split5}}) perform better than the alternatives ({\texttt {lookup-int}}, {\texttt {lookup-dac}}) a majority of the time across all queries, while being smaller in space. For query lengths greater than 10 on {\texttt {chr19}}, these approaches are faster than {\texttt {rle-string}} but slightly larger, while slower than {\texttt {RLCSA}} but smaller in size; we occupy a time/space trade-off position between these values. This is while also being much smaller than {\texttt {wt-fbb}} whose space makes it an outlier despite best speeds for various queries. 

On {\texttt {Sars-CoV2}}, our implementations perform well on queries of length 10, with {\texttt {lookup-split2}} the fastest implementation and other approaches competitive in both time/space. For query lengths greater than 10, the non-splitting approaches ({\texttt {lookup-bv}}, {\texttt {lookup-int}}, {\texttt {lookup-dac}}) perform the worst across data structures with respect to speed. With splitting approaches, we are comparable to {\texttt {rle-string}} in time but worse in space. Although again an outlier in space, {\texttt {wt-fbb}} performs fastest, with {\texttt {RLCSA}} occuping the least space with comparable speed to {\texttt {wt-fbb}}.

In terms of size/construction, we perform worse than {\texttt rle-string} across all data, but are highly competitive for {\texttt {lookup-bv}}'s space despite slower construction. For our implementations, {\texttt {lookup-bv}} is the definitive choice across results in regard to both space and construction time. When compared to {\texttt {RLCSA}}, despite being more space-efficient on {\texttt {chr19}} across {\texttt {lookup-bv}} approaches, we cannot compete on {\texttt {Sars-CoV2}} where it is a clear winner across all data structures. This motivates applying table lookup to also speed up {\texttt {RLCSA}}; however, we note adding support for $\phi$ and $\phi^{- 1}$ (thus, supporting locate) to {\texttt {RLCSA}} is still an open problem.

\begin{figure}[!htb]
\begin{minipage}{1\textwidth}
\includegraphics[width=\textwidth]{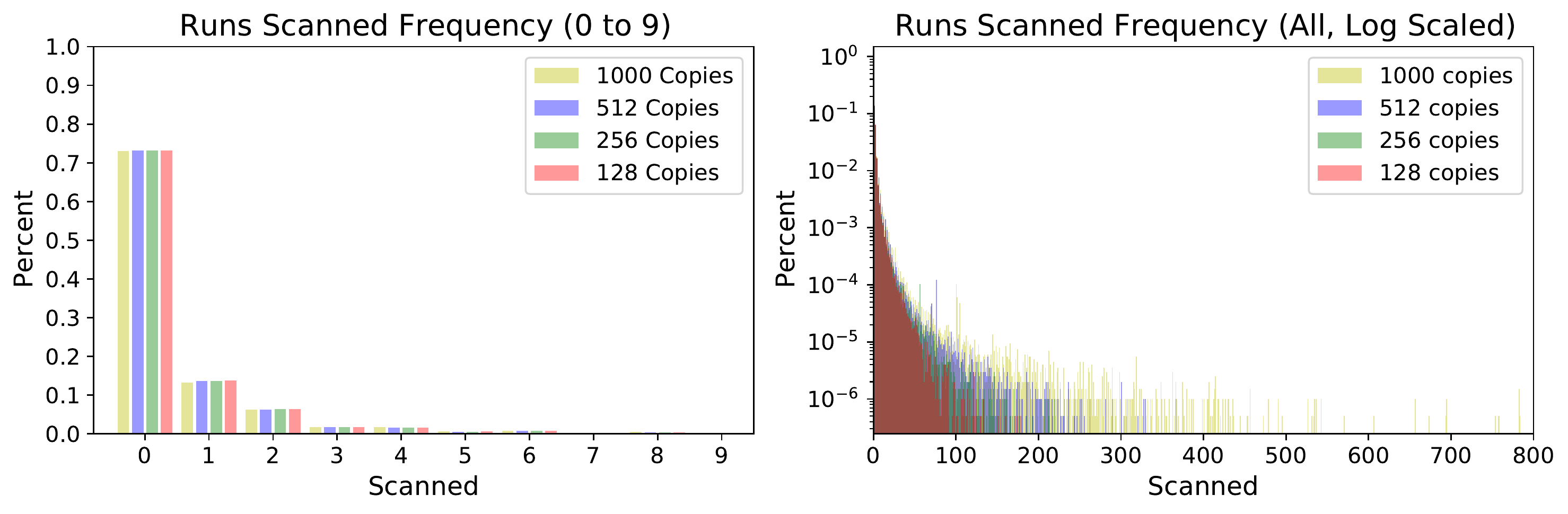}
        \end{minipage}%
        \caption{Frequencies in percentage of runs scanned for any \LF\ step across 10000 count queries of length 100 for 100, 200, 512 and 1000 copies of {\texttt {chr19}}. Plot on left is restricted only to steps scanning 0 to 9 runs; plot on right shows all scans, log scaled since the frequency of scans decreases quickly for large values.
        \label{fig:scans}}
\end{figure}

With regard to our splitting approaches, they are superior to {\texttt {lookup-bv}} for long query lengths and as $n/r$ rises. To examine the cause in terms of $n/r$ and growing text collections, we examine the number of sequential scans required across \LF\ steps during count queries of length 100 for {\texttt {chr19}} in Figure~\ref{fig:scans}. Although the distribution is similar across all copies near zero, with a majority requiring no sequential scan and most of the rest scanning very few, worst cases become both more prevalent and longer as the number of copies and $n/r$ grows. This gives further insight into the success of the splitting approaches in these instances, as bounding the maximum runs also bounds worst case sequential scans. We find this result intriguing with respect to Theorem~\ref{thm:NT21} when $n/r$ or the worst case number of scans is high. Concentrating on Nishimoto and Tabei's first result, {\texttt {lookup-bv}} performs competitively in space/time for low $n/r$ with naive run splitting as a practical alternative otherwise in our observed experiments.

\bibliography{references.bib}

\appendix

\end{document}